\documentclass[11pt,book,reqno]{amsart}
\usepackage[T1]{fontenc}
\usepackage[latin9]{inputenc}
\usepackage{amssymb}
\usepackage{color}
\usepackage{graphicx,color}
\usepackage{amsmath, amssymb, graphics}
\usepackage{qcircuit}

\makeatletter
\numberwithin{equation}{section} 
\numberwithin{figure}{section} 
  \@ifundefined{theoremstyle}{\usepackage{amsthm}}{}
  \theoremstyle{plain}
  \newtheorem{thm}{Theorem}[section]
  \theoremstyle{plain}
  
  \theoremstyle{plain}
  
  \theoremstyle{remark}
  
  \theoremstyle{remark}
  
  \theoremstyle{plain}

\smallskip
\def\<{{\langle }}
\def\>{{\rangle }}

\def\ket#1{|#1\rangle}
\def\bra#1{\langle#1|}

\smallskip
\def\<{{\langle }}
\def\>{{\rangle }}

\begin{document}

\title[Partition of three qubits ]{Partition of 3-qubits using local gates}

\author{Oscar Perdomo}
\date{\today}

\curraddr{Oscar Perdomo\\
{\bf 1.} Rigetti Computing\\
2919 Seventh Street, Berkeley, CA 94710\\
{\bf 2.} Professor\\Department of Mathematics\\
Central Connecticut State University\\
New Britain, CT 06050\\}

\email{ perdomoosm@ccsu.edu}

\begin{abstract} It is well known that local gates have smaller error than non-local gates. For this reason it is natural to define the following equivalence relation between $n$-qubit states:  $\ket{\phi_1}\sim\ket{\phi_2}$ if there 
exists a local gate $U$ such that $U\ket{\phi_1}=\ket{\phi_2}$. Since two states that differ by a local gate have the same entanglement entropy, then the entanglement entropy defines a function in the quotient space. In this paper we study this equivalence relation on (i) the set $\mathbb{R}\hbox{\rm Q}(3)$ of  3-qubit states with real amplitudes, (ii) the set  $Q\mathcal{C}$ of 3-qubit states that can be prepared with gates on the  Clifford group,  and (iii) the set $Q\mathbb{R}\mathcal{C}$ of 3-qubit states in $Q\mathcal{C}$ with real amplitudes. We show that the set $Q\mathcal{C}$ has 8460 states and the quotient space has 5 elements. We have $\frac{Q\mathcal{C}}{\sim}=\{S_{0},S_{2/3,1},S_{2/3,2},S_{2/3,3},S_{1}\}$. As usual, we will call the elements in the quotient space, orbits. We have that the orbit $S_0$ contains all the states that differ by a local gate with the state $\ket{000}$. There are 1728 states in $S_0$ and as expected, they have zero entanglement entropy.  All the states in the orbits $S_{2/3,1},S_{2/3,2},S_{2/3,3}$ have entanglement entropy $2/3$ and each one of these orbits has 1152 states. Finally, the orbit $S_{1}$ has 3456 elements and all its states have maximum entanglement entropy equal to one. We also study how the controlled not gates $CNOT(1,2)$ and $CNOT(2,3)$ act on these orbits. For example, we show that when we apply a $CNOT(1,2)$ to all the states in $S_0$, then 960 states go back to the same orbit $S_0$ and 768 states go to the orbit $S_{2/3,1}$. Similar results are obtained for  $\mathbb{R}Q\mathcal{C}$. We also show that the entanglement entropy function reaches its maximum value 1 in more than one point when acting on 
$\frac{\hbox{$\mathbb{R}$\rm Q(3)} }{\sim}$. 
\end{abstract}

\maketitle

\section{introduction}

The space of $n$-qubits is modeled by unitary vectors in $\mathbb{C}^{2^n}$ while the set of gates is modeled by the space $U(2^n)$ of square unitary matrices of dimension $2^n$. Given a subgroup $G$ of the group of unitary 2 by 2 matrices, we defined the set of local gates generated by $G$ as 

$$ L(G)=\hbox{ group generated the matrices } \{U_1\otimes\dots \otimes U_n: U_i\in G\}$$

Under the assumption that the $n$ qubits are connected by $CNOT$ gates, it is natural to define the set of all the {\it states generated by the group $G$} as the set

$$QG=\{v: v \hbox{ is prepared using CNOT gates and matrices in $L(G)$}\}$$

We say that two states $\ket{\phi_1}$ and $\ket{\phi_2}$ in $QG$ are equivalent if $\ket{\phi_1}=U\ket{\phi_2}$ for some matrix $U\in L(G)$. It is clear that a perfect understanding of the set of equivalent classes $\frac{QG}{\sim}$ turns out to be very important to prepare and deal with the states generated by the group $G$.

 Before we continue, let us denote by $U(2)$ the group of 2 by 2 unitary matrices and by  $O(2)$ the group of 2 by 2 orthogonal matrices. Let us also define the following three matrices

$$ H=\left(
\begin{array}{cc}
 \frac{1}{\sqrt{2}} & \frac{1}{\sqrt{2}} \\
 \frac{1}{\sqrt{2}} & -\frac{1}{\sqrt{2}} \\
\end{array}
\right)  \quad Z=\left(
\begin{array}{cc}
1& 0 \\
0 & -1 \\
\end{array}
\right) \, \quad  P=\left(
\begin{array}{cc}
1& 0 \\
0 & i \\
\end{array}
\right) \, .
$$

Finally, we recall that the Clifford group is the set of gates that are generated by $CNOT$ gates and  the local gates $H$ and $P$. See \cite{G}, \cite{G1}, \cite{HC} and the very recent paper \cite{T} for more properties on the Clifford group. In this paper let us denote by $\mathcal{C}$ the group of 2 by 2 matrices generated by the matrices $H$ and $P$ and by $\mathbb{R}\mathcal{C}$ the group of 2 by 2 matrices generated by $H$ and $Z$. A direct computation shows that the group $\mathbb{R}\mathcal{C}$ has 16 matrices while the set $\mathcal{C}$ has 192 elements. When $G=U(2)$ then $QG$ consists of every possible $n$-qubit state. In this case, we denote $QU(2)$ as $\mathbb{C}Q(n)$ and the quotient space $\frac{\mathbb{C}Q(n)}{\sim}$ by $\overline{\mathbb{C}Q(n)}$.  When $G=O(2)$ then $QG$ consists of every possible $n$-qubit state with all amplitudes real. In this case, we denote $QO(2)$ as $\mathbb{R}Q(n)$ and the quotient space $\frac{\mathbb{R}Q(n)}{\sim}$ by $\overline{\mathbb{R}Q(n)}$.

In the case of 2-qubit states, \cite{OA} presents a clear exposition of the quotient space $\overline{\mathbb{R}Q(2)}$. The paper shows that $\overline{\mathbb{R}Q(2)}$ is in one to one correspondence with the interval $[0,\frac{\pi}{4}]$, and if $O_t\in \overline{\mathbb{R}Q(2)}$ denotes the orbit of states represented by the number $t\in [0,\frac{\pi}{4}]$, then the states in $O_0$ have entanglement entropy 1 and  they form a pair of disjoint circles;  the states in $O_\frac{\pi}{4}$ have entanglement entropy 0 and  they form a torus; and for any $d$ between $0$ and $\frac{\pi}{4}$, the states in $O_d$ have entanglement entropy  $1- \log_2 \sqrt{\frac{(1+\sin 2 d)^{1+\sin 2 d}}{(1-\sin 2 d)^{-1+\sin 2 d}}}$ and they form a pair of disjoint tori. Moreover, the $CNOT(1,2)$ gate spreads the  states in any orbit to states in all the orbits. In other words, for any $t_1$ and $t_2$ in $[0,\frac{\pi}{2}]$

$$\{CNOT(1,2) \ket{v}: \ket{v} \in O_{t_0} \}\cap O_{t_1} \hbox{ is not the empty set }$$

As a corollary of these facts on the space $\overline{\mathbb{R}Q(2)}$ we obtain that every pair of 2-qubit states with real amplitudes can be connected using only one  $CNOT$ gate and local gates with real entries.
The paper \cite{Z} shows that  $\overline{\mathbb{C}Q(n)}$ is also in one to one correspondence with a closed interval and, again, the $CNOT(1,2)$ spreads the states in any orbit to states in all the orbits.

Moving to 3-qubit states,  the entanglement entropy function defined on the set   $\overline{\mathbb{C}Q(3)}$  reaches the maximum value $1$ at exactly one orbit, see \cite{SM}.  The first result in this paper, Theorem \ref{one}, shows that the entanglement entropy function defined on $\overline{\mathbb{R}Q(3)}$ reaches its maximum value $1$ at more than one point. We prove this by showing that if  $|\xi_1\rangle = \frac{1}{\sqrt{2}}(|000\rangle+|111\rangle)$ and $|\xi_2\rangle  = \frac{1}{2}(|001\rangle-|010\rangle+|100\rangle+|111\rangle)$, then for any local gate $U\in L(O(2))$, the system of equations coming from the equation $\ket{\xi_2}=U\ket{\xi_1}$ with $U\in L(O(2))$ has no solution. As it is known from \cite{SM}, the states $\xi_1$ and $\xi_2$ can be connected by a local gate in $L(U(2))$. We explicitly show a local gate $U$ in $L(U(2))$ that satisfies $\ket{\xi_2}=U\ket{\xi_1}$.

 In order to understand better previous result, we study/compare the quotient space $\overline{Q\mathcal{C}}=\frac{Q\mathcal{C}}{\sim}$ and the quotient space $\overline{\mathbb{R}Q\mathcal{C}}=\frac{Q\mathbb{R}\mathcal{C}}{\sim}$. We check that  $Q\mathcal{C}$ has 8460 states and it gets partitioned into $5$ orbits to form the space $\overline{Q\mathcal{C}}$. The action of the controlled not gates $CNOT(1,2)$ and $CNOT(2,3)$ is explained in Figure \ref{cg}.
\begin{figure}[hbtp]\label{cg}
\begin{center}\includegraphics[width=.65\textwidth]{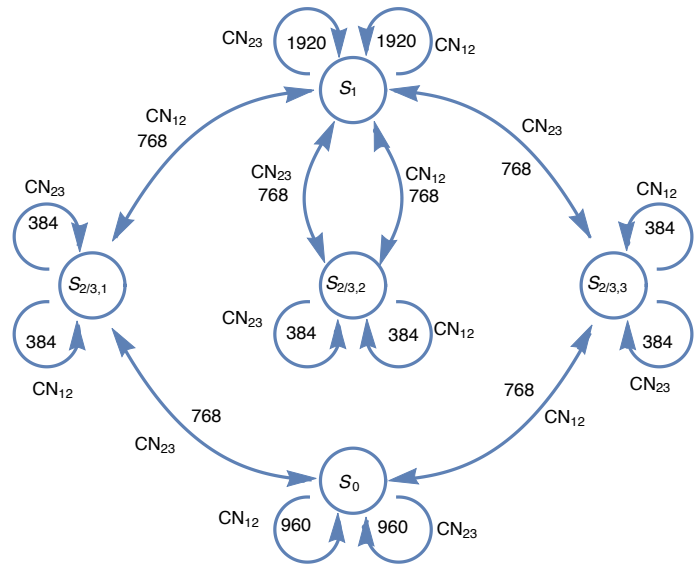}
\caption{$\overline{Q\mathcal{C}}$ has 5 orbits. The number by the side of the arrows indicate the number of states that are moved  from one orbit to another by the given  $CNOT$ gate.} \label{cg}
\end{center}
\end{figure}

The set $Q\mathbb{R}\mathcal{C}$ has 480 3-qubit states and it gets partitioned into $6$ orbits to form the space $\overline{Q\mathbb{R}\mathcal{C}}$. The action of the controlled not gates $CNOT(1,2)$ and $CNOT(2,3)$ is explained in Figure \ref{rg}.  For counting purposes, we are treating states that are related by global phase as distinct. For example, the states $\ket{\phi}$ and $-\ket{\phi}$ considered to be two states. Of Course they are in the same orbit.

\begin{figure}[hbtp]
\begin{center}\hskip1cm\includegraphics[width=.65\textwidth]{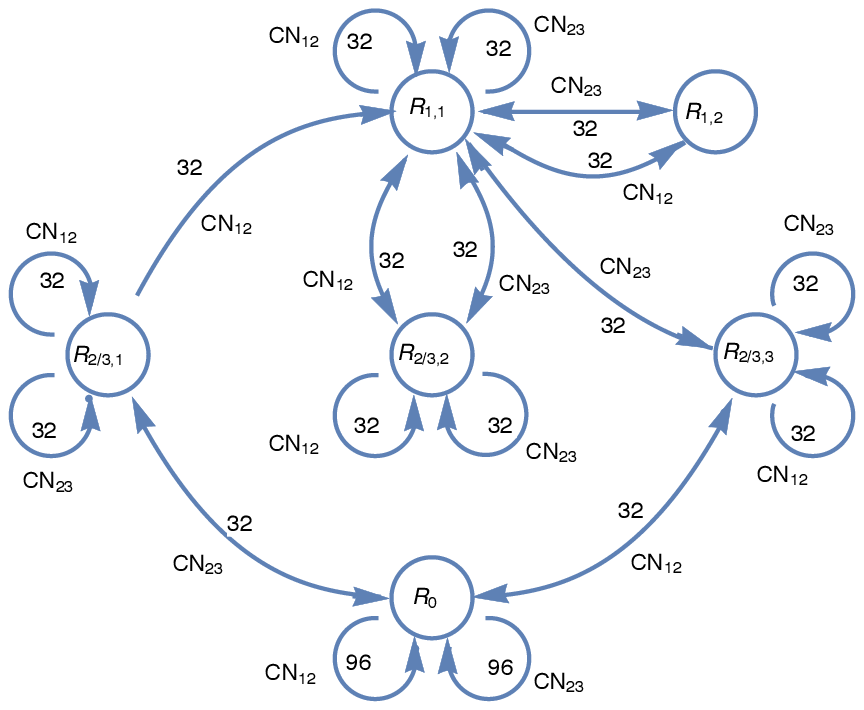}  
\caption{Action of the controlled not gates on the space $Q\mathbb{R}\mathcal{C}$. }\label{rg}
\end{center}
\end{figure}

\section{Proof on the main results}

\subsection{On the maximum of the entanglement entropy function on the set $\overline{\mathbb{R}Q(3)}$ }

It is well known that the maximum of the entanglement entropy function on the 3-qubit states $\mathbb{C}Q(3)$ is 1 and that  any two states with entanglement entropy $1$ can be connected by a local gate in $L(U(2))$, \cite{SM}.  As a consequence we have that the entanglement entropy function defined on the space $\overline{\mathbb{C}Q(3)}$ reaches its maximum  at only one point.
The following Theorem shows that the maximum of the entanglement entropy function  in the space $\overline{\mathbb{R}Q(3)}$ is also 1 but this maximum  happens at more than one point.

\begin{thm}\label{one} No gate of the form $A_1\otimes A_2\otimes A_3$ with $A_i$ an orthogonal 2 by 2  matrix  transforms the 3-qubit $\ket{\xi_2}=\frac{1}{2}(|001\rangle-|010\rangle+|100\rangle+|111\rangle)$ into the qubit  $\ket{\xi_1}=\frac{1}{\sqrt{2}}(|000\rangle+|111\rangle)$. Moreover,  a direct verification shows that if,
$$ A_1=\left(
\begin{array}{cc}
 \frac{e^{-\frac{1}{130} (341 i)}}{\sqrt{2}} & -\frac{i e^{-\frac{1}{130} (341 i)}}{\sqrt{2}} \\
 \frac{e^{\frac{341 i}{130}}}{\sqrt{2}} & \frac{i e^{\frac{341 i}{130}}}{\sqrt{2}} \\
\end{array}
\right),\quad
A_2= \left(
\begin{array}{cc}
 \frac{i e^{-\frac{1}{65} (18 i)}}{\sqrt{2}} & -\frac{e^{-\frac{1}{65} (18 i)}}{\sqrt{2}} \\
 \frac{e^{\frac{18 i}{65}}}{\sqrt{2}} & -\frac{i e^{\frac{18 i}{65}}}{\sqrt{2}} \\
\end{array}
\right)$$
and 

$$A_3=\left(
\begin{array}{cc}
 \frac{e^{\frac{29 i}{10}}}{\sqrt{2}} & -\frac{i e^{\frac{29 i}{10}}}{\sqrt{2}} \\
 -\frac{i e^{-\frac{1}{10} (29 i)}}{\sqrt{2}} & \frac{e^{-\frac{1}{10} (29 i)}}{\sqrt{2}} \\
\end{array}
\right)$$

then, $\ket{\psi_1}=A_1\otimes A_2 \otimes A_3 \ket{\psi_2}$
\end{thm}

\begin{proof} Let us consider a matrix $U$ of the form $A_1\otimes  A_2 \otimes A_3$, where,

$$A_1=\begin{pmatrix}
    a_1 & a_2  \\
    -a_2 & a_1
  \end{pmatrix} ,\quad A_2= \begin{pmatrix}
    a_3 & a_4  \\
    -a_4 & a_3
  \end{pmatrix} ,\quad A_3=\begin{pmatrix}
    a_5 & a_6  \\
    -a_6 & a_5
  \end{pmatrix}$$
  
 where the equations
 
 \begin{eqnarray}\label{eq1}
 a_1^2+a_2^2=1,\quad  a_3^2+a_4^2=1,\quad  a_5^2+a_6^2=1
 \end{eqnarray}
 
 are safisfied.  We have that the matrix $U$ cannot send the state  $|\psi_2\rangle  = \frac{1}{2}(|001\rangle-|010\rangle+|100\rangle+|111\rangle)$ to the state $|\psi_1\rangle = \frac{1}{\sqrt{2}}(|000\rangle+|111\rangle)$ because the system of equations

 \begin{eqnarray}\label{eq2}
 \bra{|0\rangle}U|\psi_2\rangle- \bra{|7\rangle}U|\psi_2\rangle=0,\quad  \bra{|j\rangle}U|\psi_2\rangle=0\quad \hbox{j=1\dots 6}
 \end{eqnarray}
 
along with the Equation (\ref{eq1}) do not have a solution. Here $\ket{0}=\ket{000},\dots , \ket{7}=\ket{111}$. Due to the symmetries, two of the 7 equations in (\ref{eq2}) are repeated. We only have the following 5 equations

\begin{eqnarray*}
& -\frac{1}{2} a_1 a_3 a_5-\frac{1}{2} a_2 a_4 a_5+\frac{1}{2} a_2 a_3 a_6-\frac{1}{2} a_1 a_4 a_6=0\\
 &\frac{1}{2} a_1 a_3 a_5+\frac{1}{2} a_2 a_4 a_5-\frac{1}{2} a_2 a_3 a_6+\frac{1}{2} a_1 a_4 a_6=0\\
 &-\frac{1}{2} a_2 a_3 a_5+\frac{1}{2} a_1 a_4 a_5-\frac{1}{2} a_1 a_3 a_6-\frac{1}{2} a_2 a_4 a_6=0\\
 &\frac{1}{2} a_2 a_3 a_5-\frac{1}{2} a_1 a_4 a_5+\frac{1}{2} a_1 a_3 a_6+\frac{1}{2} a_2 a_4 a_6=0\\
 &-\frac{1}{2} a_1 a_3 a_5+\frac{1}{2} a_2 a_3 a_5-\frac{1}{2} a_1 a_4 a_5-\frac{1}{2} a_2 a_4 a_5+\frac{1}{2} a_1 a_3 a_6+\frac{1}{2} a_2 a_3 a_6-\frac{1}{2} a_1 a_4 a_6+\frac{1}{2} a_2 a_4 a_6=0
\end{eqnarray*}

A direct computation shows that a Grobner basis of this system of equations above is $\{1\}$, and therefore the system does not have a solution and the matrix $U$ does not exist. This takes care of the case when all the matrices $A_1$, $A_2$ and $A_3$ have determinant 1. A similar proof can be done for the other 7 cases coming from the other possibilities for the determinant of $A_i$.

\end{proof}

\subsection{The space $\overline{Q\mathcal{C}}$ } In this section we describe the partition induced by the equivalence relation on the set of states that can be prepared with Clifford gates. 

\begin{thm} Given a set $B$ let us denote $|B|$ the number of elements of $B$. The gubgroup $\mathcal{C}\subset U(2)$  has 192 elements and for the case of 3-qubits, $L(\mathcal{C})\subset U(8)$ has 110592 elements. Moreover, the set $Q\mathcal{C}\subset \mathbb{C}^8$ has  8640 states and if we define

$S_0=S_{0,1}=\{U(1,0,0,0,0,0,0,0)^T:U\in L(\mathcal{C})\}$

$S_{2/3,1}=\{U(-1/2,-1/2,-1/2,1/2,0,0,0,0)^T:U\in L(\mathcal{C})\}$

$S_{2/3,2}=\{U(-1/2,-1/2,0,0,-1/2,1/2,0,0)^T:U\in L(\mathcal{C})\}$

$S_{2/3,3}=\{U(-1/2,-1/2,0,0,0,0,-1/2,-1/2)^T:U\in L(\mathcal{C})\}$

$S_1=S_{1,1}=\{U(-1/2,-1/2,0,0,0,0,-1/2,1/2)^T:U\in L(\mathcal{C})\}$

Then all the $S_{i,j}$ are disjoint and their union is the set $Q\mathcal{C}$. The entanglement entropy of all the states in $S_{i,j}$ is $i$ and $S_0$ has 1728 states,  $S_{2/3,1}$ has 1152 states,  $S_{2/3,2}$ has 1152 states,  $S_{2/3.3}$ has 1152 states,  $S_1$ has 3456 states. Additionally, if we define

$$cn_{ij}S=\{CNOT(i,j)\ket{\phi}:\ket{\phi}\in S \}$$ 

then,
\begin{enumerate}
\item
For $(i,j)=(1,2),\, (2,3)$ or $(1,3)$, and $k=1,2,3$ we have  $|cn_{ij}S_0\cap S_0|=960$,  $|cn_{ij}S_{2/3,k}\cap S_{2/3,k}|=384$, and  $|cn_{ij}S_1\cap S_1|=1920$.

\item
$|cn_{23}S_0\cap S_{2/3,1}|=768$, $|cn_{12}S_0\cap S_{2/3,3}|=768$, $|cn_{13}S_0\cap S_{2/3,2}|=768$.

\item

$|cn_{12}S_{2/3,1}\cap S_{1}|=768$, $|cn_{13}S_{2/3,1}\cap S_{1}|=768$, $|cn_{23}S_{2/3,2}\cap S_{1}|=768$, $|cn_{12}S_{2/3,2}\cap S_{1}|=768$, $|cn_{23}S_{2/3,3}\cap S_{1}|=768$, $|cn_{13}S_{2/3,3}\cap S_{1}|=768$ 
\end{enumerate}

\end{thm} 

\begin{proof}
The proof is made by searching by exhaustion for all the states and local gates. An explanation of why the group $\mathcal{C}$ has 192 can be found in \cite{T}.
\end{proof}

\subsection{The space $\overline{\mathbb{R}Q\mathcal{C}}$ } This section gives some properties of the states in $\mathbb{R}Q\mathcal{C}$.

\begin{thm} The gubgroup $\mathbb{R}\mathcal{C}\subset U(2)$  has 16 elements and for the case of 3-qubits, $L(\mathbb{R}\mathcal{C})\subset U(8)$ has 1024 elements. In general for the case of $n$-qubits, $|L(\mathbb{R}\mathcal{C})|=2^{3n+1}$. Moreover, the set $\mathbb{R}Q\mathcal{C}\subset \mathbb{C}^8$ has  480 states and they are exactly those states $Q\mathcal{C}\subset \mathbb{C}^8$ with all its amplitudes real. If we define,

$R_0=R_{0,1}=\{U(1,0,0,0,0,0,0,0)^T:U\in L(\mathcal{C})\}$

$R_{2/3,1}=\{U(-1/2,-1/2,-1/2,1/2,0,0,0,0)^T:U\in L(\mathcal{C})\}$

$R_{2/3,2}=\{U(-1/2,-1/2,0,0,-1/2,1/2,0,0)^T:U\in L(\mathcal{C})\}$

$R_{2/3,3}=\{U(-1/2,-1/2,0,0,0,0,-1/2,-1/2)^T:U\in L(\mathcal{C})\}$

$R_{1,1}=\{U(-1/2,-1/2,0,0,0,0,-1/2,1/2)^T:U\in L(\mathcal{C})\}$

$R_{1,2}=\{U(-1/2,0,0,-1/2,0,-1/2,1/2,0)^T:U\in L(\mathcal{C})\}$

Then, all the $R_{i,j}$ are disjoint and their union is the set $\mathbb{R}Q\mathcal{C}$. The entanglement entropy of all the elements in $R_{i,j}$ is $i$ and $R_0$ has 128 states,  $R_{2/3,1}$ has 64 states, as well as  $R_{2/3,2}$ and  $R_{2/3.3}$,   $R_{1,1}$ has 128 states and $R_{1,2}$ has 32 states.

\begin{enumerate}
\item
For $(i,j)=(1,2),\, (2,3)$ or $(1,3)$, and $k=1,2,3$ we have  $|cn_{ij}R_0\cap R_0|=96$,  $|cn_{ij}R_{2/3,k}\cap R_{2/3,k}|=32$, $|cn_{ij}R_{1,1}\cap R_{1,1}|=32$ and  $|cn_{ij}R_{1,1}\cap R_{1,2}|=32$
\item

$|cn_{12}R_0\cap R_{2/3,3}|=32$, $|cn_{23}R_0\cap R_{2/3,1}|=32$, $|cn_{13}R_0\cap R_{2/3,2}|=32$.

\item

$|cn_{12}R_{2/3,1}\cap R_{1,1}|=32$, $|cn_{13}R_{2/3,1}\cap R_{1,1}|=32$, $|cn_{23}R_{2/3,2}\cap R_{1,1}|=32$, $|cn_{12}R_{2/3,2}\cap R_{1,1}|=32$, $|cn_{23}R_{2/3,3}\cap R_{1,1}|=32$, $|cn_{13}R_{2/3,3}\cap R_{1,1}|=32$ 
\end{enumerate}

\end{thm} 

\begin{proof}

To prove that in  the case of $n$-qubits, $|L(\mathbb{R}\mathcal{C})|=2^{3n+1}$ we argue by induction and we use the fact that $V\otimes U$ and  $(-V)\otimes (-U)$. The proof of the other statements are made 
 by exhaustion on all the states and local gates being considered.
 
\end{proof}

\section{Conclusions}

In the paper \cite{Z} the authors define the distance between two states as the number of $CNOT$ gates needed to transform one state to the other. They show that for 3-qubits, any state is within a distance of $3$ to the state $\ket{000}$. Also they prove that the distance between any pair of states is less than or equal to 4.  They leave as an open question to find out if four is the maximum distance between two 3-qubit states. 

Figures \ref{cg} and \ref{rg} allow us to conclude that using the line topology, this is, using only $CNOT(1,2)$ and $CNOT(2,3)$ and local gates form $L(\mathcal{C})$, we have that the maximum distance between any pair of states in $Q\mathcal{C}$ is 3 and, if we use the all-to-all topology, this is, if we consider all the $CNOT$ gates, then the maximum distance is 2. Likewise, using the line topology with local gates form $L(\mathbb{R}\mathcal{C})$, the maximum distance between any pair of states in $Q\mathcal{C}$ is 3 and  if we use the 
all-to-all topology with local gates in $L(\mathbb{R}\mathcal{C})$, the maximum distance remain being 3.

\section{Acknowledgments}

The author would like to thank Eric Peterson, Marcus P. da Silva, Alejandro Perdomo-Ortiz and Vicente Leyton for suggestions and feedback on an early draft of this paper.

\end{document}